\documentclass[11pt]{article}

\usepackage[utf8]{inputenc}

\usepackage[margin=1in]{geometry}

\usepackage{graphicx}
\usepackage{tabu}
\usepackage{cite}
\usepackage{latexsym}
\usepackage{amssymb}
\usepackage[in]{fullpage}
\usepackage{amsmath, amsfonts, amsthm}
\usepackage{graphicx}
\usepackage[noend]{algpseudocode}
\usepackage{algorithm}
\usepackage{comment}

\usepackage{booktabs}
\usepackage{array}
\usepackage{paralist}
\usepackage{verbatim}
\usepackage{subfig}

\usepackage{fancyhdr}
\usepackage{algorithm}

\newtheorem{ba}{Black Box Algorithm}
\newtheorem{definition}{Definition}
\newtheorem{theorem}{Theorem}
\newtheorem{lemma}{Lemma}
\newtheorem{corollary}{Corollary}
\newtheorem{problem}{Problem}

\title{Approximating Subadditive Hadamard Functions on Implicit Matrices}

\author{
Vladimir Braverman\thanks{Department of Computer Science, Johns Hopkins University; \texttt{vova@cs.jhu.edu}, \texttt{gregvorsanger@jhu.edu}}
\and
Alan Roytman\thanks{School of Computer Science, Tel-Aviv University; \texttt{alan.roytman@cs.tau.ac.il}}
\and
Gregory Vorsanger\footnotemark[1]
}
\date{}

\begin{document}
\maketitle

\begin{abstract}
An important challenge in the streaming model is to maintain small-space approximations of
entrywise functions performed on a matrix that is generated by the outer product of two vectors given as a stream.
In other works, streams typically define matrices in a standard way via a sequence of updates, as in the
work of Woodruff \cite{W14} and others. We describe the matrix formed by the outer product, and other matrices that do
not fall into this category, as implicit matrices. As such, we consider the general problem of computing over
such implicit matrices with Hadamard functions, which are functions applied entrywise on a matrix. In this paper,
we apply this generalization to provide new techniques for identifying independence between two
vectors in the streaming model.  The previous state of the art algorithm of Braverman and Ostrovsky~\cite{BO10}
gave a $(1 \pm \epsilon)$-approximation for the $L_1$ distance between the product and joint distributions,
using space $O(\log^{1024}(nm) \epsilon^{-1024})$, where $m$ is the length of the stream and $n$ denotes the
size of the universe from which stream elements are drawn.  Our general techniques include the $L_1$ distance
as a special case, and we give an improved space bound of $O(\log^{12}(n) \log^{2}({nm \over \epsilon})\epsilon^{-7})$.
\end{abstract}

\section{Introduction}

Measuring Independence is a fundamental statistical problem that is well studied in computer science.
Traditional non-parametric methods of testing independence over empirical data usually require space complexity that is
polynomial in either the support size or input size. With large datasets, these space requirements may be impractical,
and designing small-space algorithms becomes desirable.

Measuring independence is a classic problem in the field of statistics (see Lehmann~\cite{lehmann2006testing}) as
well as an important problem in databases. Further, the process of reading in a two-column database table can be
viewed as a stream of pairs. Thus, the streaming model is a natural choice when approximating pairwise independence
as memory is limited. Indeed, identifying correlations between database columns by measuring the level of independence between
columns is of importance to the database and data warehouse community
(see, e.g., \cite{poosala1997selectivity} and \cite{kimball2004data}, respectively).

In this paper we provide new techniques for measuring independence between two vectors in the streaming model
and present new tools to expand existing techniques. The topic of independence was first studied in the streaming model
by Indyk and McGregor~\cite{IM08} where the authors gave an optimal algorithm for approximating the $L_2$ distance
between the product and joint distributions of two random variables which generate a stream. In their work, they provided a sketch that
is pairwise independent, but not $4$-wise independent, so analysis similar to that of Alon, Matias, and Szegedy~\cite{ams} cannot be applied directly.
This work was continued by Braverman and Ostrovsky~\cite{BO10}, where the authors considered comparing among a stream of $k$-tuples and provided
the first $(1\pm \epsilon)$-approximation for the $L_1$ distance between the product and joint distributions.  Their algorithm is
currently the best known space bound, and uses $O({1 \over \epsilon^{1024}} \log^{1024}({nm}))$ space for $k=2$,
where $m$ is the length of the stream and $n$ denotes the size of the universe from which stream elements are drawn. We present new methods,
in the form of a general tool, that enable us to improve this bound to
$O({ 1 \over \epsilon^7}\log^{12}(n) \log^{2}({nm \over \epsilon}))$. In previous works, a central challenge
has been maintaining an approximation of the matrix that is generated by the outer product of the two streaming vectors.
As such, we consider computing functions on such an implicit matrix. While, matrices have been studied previously in the streaming model (e.g., \cite{W14}), note that
we cannot use standard linear sketching techniques, as the entries of the matrix are given implicitly and thus these methods do not apply directly.

Generalizing this specific motivating example, we consider the problem of obtaining a $(1 \pm \epsilon)$-approximation of the $L_1$ norm of the matrix
$g[A]$, where $g[A]$ is the matrix $A$ with a function $g$ applied to it entrywise.
Such mappings $g$ are called \emph{Hadamard} functions (see~\cite{guillot2015complete,horntopics}).
Note that we sometimes abuse notation and apply the function $g$ to scalar values instead of
matrices (e.g., $g(a_{ij})$ where $a_{ij}$ is the $(i,j)^{th}$ entry in matrix $A$).
We require the scalar form of function $g$ to be even, subadditive, non-negative, and zero at the origin.
We show that, given a blackbox $r(n)$-approximation of
$\|g[A]\|_1 = \sum_i\sum_jg(a_{ij})$ (where $a_{ij}$ is the $(i,j)^{th}$ entry in matrix $A$) and a blackbox
$(1\pm \epsilon)$-approximation of the aggregate of $g$ applied entrywise to a vector obtained by summing over
all rows, we are able to improve the $r(n)$-approximation to a $(1 \pm \epsilon)$-approximation (where
$r(n)$ is a sufficiently large monotonically increasing function of $n$).
Hence, we give a reduction for any such function $g$.  Our reduction can be applied as long as such
blackbox algorithms exist.  An interesting special case of our result is when the matrix is defined by the $L_1$
distance between the joint and product distributions,
which corresponds to measuring independence in data streams.  Such algorithms are known for $L_1$, but not
for $L_p$ for $0 < p < 1$.  If such algorithms for the $L_p$ distance were to be designed,
our reductions work and can be applied.  Note that, while there are a variety of ways to compute distances
between distributions, the $L_p$ distance is of particular significance as evidenced in~\cite{stable}.

\subsection*{Motivating Problem}\label{motivating}

We begin by presenting our motivating problem, which concerns (approximately) measuring the distance between
the product and joint distributions of two random variables.  That is, we attempt to quantify how close
two random variables $X$ and $Y$ over a universe $[n] = \{1,\ldots,n\}$ are to being independent.
There are many ways to measure the distance between distributions, but we focus on the $L_1$ distance.
Recall that two random variables $X$ and $Y$ are independent if, for every $i$ and $j$,
we have $\Pr[X = i \wedge Y = j] = \Pr[X = i]\Pr[Y = j]$.
In our model, we have a data stream $D$ which is presented as a sequence of $m$ pairs
$a_1 = (i_1,j_1), a_2 = (i_2,j_2),\ldots,a_m = (i_m,j_m)$.  Each pair $a_k = (i_k,j_k)$
consists of two integers taken from the universe $[n]$.

Intuitively, we imagine that the two random variables $X$ and $Y$ over the universe $[n]$ generate these pairs, and in particular,
the frequencies of each pair $(i,j)$ define an empirical joint distribution, which is the fraction
of pairs that equal $(i,j)$.  At the same time, the stream also defines the empirical marginal
distributions $\Pr[X=i], \Pr[Y=j]$, namely the fraction of pairs of the form $(i,\cdot)$ and $(\cdot,j)$, respectively.
We note that, even if the pairs are actually generated from two independent sources, it may not be the case that the
empirical distributions reflect this fact, although for sufficiently long streams the joint distribution should approach
the product of the marginal distributions for each $i$ and $j$.  This fundamental problem has received considerable
attention within the streaming community, including the works of~\cite{IM08,BO10}.
\begin{problem}\label{motprob}
Let $X$ and $Y$ be two random variables defined by the stream of $m$ pairs $a_1 = (i_1,j_1),\ldots,a_m = (i_m,j_m)$,
where each $i_k,j_k \in [n]$ for all $k$.  Define the frequencies $f_i = |\{k: a_k = (i,\cdot)\}|$ and $f_j = |\{k: a_k = (\cdot,j)\}|$
(i.e., the frequency with which $i$ appears in the first coordinate and $j$ appears in the second coordinate, respectively).
Moreover, let $f_{ij} = |\{k : a_k = (i,j)\}|$ be the frequency with which the pair $(i,j)$ appears in the stream.
This naturally defines the joint distribution $Pr[X=i \wedge Y=j] = \frac{f_{ij}}{m}$ and the product of the marginal distributions
$Pr[X=i] Pr[Y=j] = \frac{f_i f_j}{m^2}$. The $L_1$ distance between the product and joint distributions is given by:
$$ \sum\limits_{i =1}^n \sum\limits_{j = 1}^n \left|\frac{f_{ij}}{m}  - \frac{f_i f_j}{m^2}\right|.$$
\end{problem}

If $X$ and $Y$ are independent, we should expect this sum to be close to $0$, assuming the stream is sufficiently long.
As a generalization to this problem, we
can view the $n^2$ values which appear in the summation as being implicitly
represented via an $n \times n$ matrix $A$, where entry $a_{ij} = \left|\frac{f_{ij}}{m}  - \frac{f_i f_j}{m^2}\right|$.
For the motivating problem, this matrix is given implicitly as it is not given up front and changes over time
according to the data stream (each new pair in the stream may change a particular entry in the matrix).
However, one can imagine settings in which these entries are defined through other means.  In practice,
we may still be interested in computing approximate statistics over such implicitly defined matrices.

\subsection*{Contributions and Techniques}

Our main contributions in this paper make progress on two important problems:

\begin{enumerate}
\item For any subadditive, even Hadamard function $g$ where $g$ is non-negative and $g(0) = 0$, given an implicitly defined $n \times n$ matrix $A$ with
entries $a_{ij}$, let $g[A]$ be the matrix where the $(i,j)^{th}$ entry is $g(a_{ij})$.  We are the first to provide a general reduction
framework for approximating $\|g[A]\|_1 = \sum_{i=1}\sum_{j=1}^n g(a_{ij})$ to within a $(1 \pm \epsilon)$-factor with constant
success probability.  More formally, suppose we have two blackbox algorithms with the following guarantees.
One blackbox algorithm operates over the implicit matrix $A$ and provides a very good ($ \approx 1 \pm \epsilon$) approximation
to $\|g[JA]\|_1 = \sum_{j=1}^n g(\sum_{i=1}^n a_{ij})$ except with inverse polylogarithmic probability, where $J = (1,\ldots,1)$ is the row vector of dimension $n$ with every entry equal to $1$.  The second blackbox algorithm operates over the implicit matrix $A$ and solves the problem we wish to solve
(i.e., approximating $\|g[A]\|_1$) with constant success probability, although it does so with a multiplicative approximation ratio of
$r(n)$ (which may be worse than $(1 \pm \epsilon)$ in general).  We show how to use these two blackbox algorithms and construct an algorithm
that achieves a $(1 \pm \epsilon)$-approximation of $\|g[A]\|_1$.  If $S_1,S_2$ denote the space used by the first and second blackbox algorithms,
respectively, then our algorithm uses space $O\left(\frac{r^4(n)log^8(n)}{\epsilon^5} \cdot (\log^3(n) + S_1 + \log(n) \cdot  S_2)\right)$.
We state this formally in Theorem~\ref{thm:main}.

\item Given the contribution above, it follows that setting $g(x) = |x|$ solves Problem~\ref{motprob}, namely the problem
of measuring how close two random variables are to being independent, as long as such blackbox algorithms exist.  In particular,
the work of Indyk~\cite{stable} provides us with the first blackbox algorithm, and the work of~\cite{IM08} provides us with
the second blackbox algorithm for this choice of $g$.  Combining these results, we improve over the previous state of the art
result of Braverman and Ostrovsky~\cite{BO10} and give improved bounds for measuring independence of random variables in the streaming model
by reducing the space usage from $O\left((\frac{\log(nm)}{\epsilon})^{1024}\right)$ to $O\left({1 \over \epsilon^7}\log^{12}(n)\log^2(\frac{nm}{\epsilon})\right)$
(see Table~\ref{tab:results}).
\end{enumerate}

{\renewcommand{\arraystretch}{1.6}
\begin{table}[H]
\centering

    \begin{tabular}{| c | c | c | c | c |}
    \hline
    Previous Work & $L_1$ approximation & Memory \\ \hline
    IM08 \cite{IM08} & $\log(n)$ & $O\left(\frac{1}{\epsilon^2} \log\left(\frac{nm}{\epsilon}\right) \log\left(\frac{m}{\epsilon}\right) \right)$\\ \hline
    BO10\footnotemark[1] \cite{BO10}& $(1 \pm \epsilon)$ & $O\left(\left(\frac{\log(nm)}{\epsilon}\right)^{1024}\right)$  \\ \hline
    Our Result & $(1 \pm \epsilon)$ & $O\left(\frac{1}{\epsilon^7}\log^{12}(n) \log^{2}\left({nm \over \epsilon}\right)\right)$.  \\
    \hline
    \end{tabular}

\caption{Comparing Approximation Ratios and Space Complexity}
\label{tab:results}
\end{table}
}

\footnotetext[1]{The paper of~\cite{BO10} provides a general bound for the $L_1$ distance for $k$-tuples, but we provide analysis for
pairs of elements, $k=2$, in this paper. The bound in the table is for $k=2$.}
Examples of such Hadamard functions which are subadditive, even, non-negative, and zero at the origin
include $g(x) = |x|^p$, for any $0 < p \leq 1$.
Note that our reduction in the first item can only be applied to solve the problem of approximating $\|g[A]\|_1$ if
such blackbox algorithms exist, but for some functions $g$ this may not be the case.  As a direct example of
the tools we present, we give a reduction for computing the $L_p$ distance for $0 < p < 1$
between the joint and product distributions in the streaming model (as this function is even and subadditive).
However, to the best of our knowledge, such blackbox algorithms do not exist for computing the $L_p$ distance.
Thus, as a corollary to our main result, the construction of such blackbox algorithms that are space efficient would immediately yield
an algorithm that measures independence according to the $L_p$ distance that is also space efficient.

Our techniques leverage concepts provided in \cite{BO10,IM08} and manipulates them to allow
them to be combined with the Recursive Sketches data structure \cite{recursive1} to gain a
large improvement compared to existing bounds.
Note that we cannot use standard linear sketching techniques because the entries of the matrix are given implicitly.
Moreover, the sketch of Indyk and McGregor~\cite{IM08} is pairwise independent, but not $4$-wise independent.  Therefore,
we cannot apply the sketches of~\cite{ams,IM08} directly. We first present an algorithm,
independent of the streaming model, for finding heavy rows of a matrix norm given an arbitrary even subadditive
Hadamard function $g$.  We then apply the Recursive Sum algorithm from~\cite{recursive1} on top of our heavy rows algorithm to obtain
our main result.

\subsection{Related Work}

In their seminal 1996 paper Alon, Matias and Szegedy\cite{ams} provided an optimal space
approximation for $L_2$. A key technical requirement of the sketch is the assumption of $4$-wise
independent random variables.  This technique is the building block for measuring the independence
of data streams using $L_2$ distances as well.

The problems of efficiently testing pairwise, or $k$-wise, independence
were considered by Alon, Andoni, Kaufman, Matulef, Rubinfeld and Xie \cite{alon2007testing}; Alon,
Goldreich and Mansour \cite{alon2003almost}; Batu, Fortnow, Fischer, Kumar, Rubinfeld and White \cite{BatuTesting2001}; Batu, Kumar and Rubinfeld \cite{batu2004sublinear}; Batu, Fortnow, Rubinfield, Smith and White \cite{BatuIndependence} and \cite{batu2013testing}. They addressed the problem of minimizing the number of samples needed to obtain a sufficient approximation, when the joint distribution is accessible through a
sampling procedure.

In their 2008 work, Indyk and McGregor \cite{IM08} provided exciting results for identifying the correlation of two streams, providing an optimal bound for determining the $L_2$ distance between the product and joint distributions of two random variables. 

In addition to the $L_2$ result, Indyk and McGregor presented a $\log(n)$-approximation for the $L_1$ distance. This bound was improved
to a $(1 \pm \epsilon)$-approximation in the work of Braverman and Ostrovsky \cite{BO10} in which they provided a bound of
$O({1 \over \epsilon^{1024}} \log^{1024}(nm))$ for pairs of elements. Further, they gave bounds for the comparison of multiple
streaming vectors and determining $k$-wise relationships for $L_1$ distance. Additionally, in \cite{braverman2010ams4} Braverman
et al. expanded the work of \cite{IM08} to $k$ dimensions for $L_2$. While our paper only addresses computation on matrices
resulting from pairwise comparison ($k=2$), we believe the techniques presented here can be expanded to tensors, (i.e., when stream
elements are $k$-tuples), similarly to \cite{braverman2010ams4}.   Recently, McGregor and Vu \cite{mcgregorevaluating} studied
a related problem regarding Bayesian networks in the streaming model.

\emph{Statistical Distance}, $L_1$, is one of the most fundamental metrics for measuring the similarity of two distributions. It has been the metric of choice in many of the above testing papers, as well as others such as Rubinfeld and Servedio \cite{Rubinfeld05testingmonotone}; Sahai and Vadhan \cite{Sahai2003Complete}. As such, a main focus of this work is improving bounds for this measure in the streaming model.

\section{Problem Definition and Notation}\label{SFandBA}

In this paper we focus on the problem of approximating even, subadditive, non-negative Hadamard functions
which are zero at the origin on implicitly defined matrices (e.g., the streaming model implicitly defines matrices for
us in the context of measuring independence).  The main problem we study in this paper is the following:

\begin{problem}\label{genprob}
Let $g$ be any even, subadditive, non-negative Hadamard function such that $g(0) = 0$.  Given any implicit matrix $A$,
for any $\epsilon > 0$, $\delta > 0$, output a $(1 \pm \epsilon)$-approximation of $\|g[A]\|_1$ except with probability $\delta$.
\end{problem}

We now provide our main theorem, which solves Problem~\ref{genprob}.
\begin{theorem}\label{thm:main}
Let $g$ be any even, subadditive, non-negative Hadamard function $g$
where $g(0) = 0$, and fix $\epsilon > 0$.  Moreover, let $A$ be an arbitrary matrix, and $J$ be the all $1$'s row vector $J = (1,\ldots,1)$
of dimension $n$.  Suppose there are two blackbox algorithms with the following properties:

\begin{enumerate}

\item Blackbox Algorithm $1$, for all $\epsilon' > 0$, returns a $(1\pm\epsilon')$-approximation of $\|g[JA]\|_1$, except with probability $\delta_1$.

\item Blackbox Algorithm $2$ returns an $r(n)$-approximation of $\|g[A]\|_1$, except with probability $\delta_2$ (where $r(n)$ is a sufficiently
large monotonically increasing function of $n$).

\end{enumerate}

Then, there exists an algorithm that returns a $(1\pm\epsilon)$-approximation of $\|g[A]\|_1$, except with constant probability. If Blackbox Algorithm $1$
uses space $SPACE1(n,\delta_1,\epsilon')$, and Blackbox Algorithm $2$ uses space $SPACE2(n,\delta_2)$, the
resulting algorithm has space complexity
$$O\left( {1 \over \epsilon^5} r^4(n)\log^{10}(n) + {1 \over \epsilon^5} r^4(n)\log^8(n)SPACE1(n,\delta_1,\epsilon') +  {1 \over \epsilon^5} r^4(n)\log^9(n)SPACE 2(n,\delta_2)\right),
$$
where $\epsilon' = \frac{\epsilon}{2}$, $\delta_1$ is a small constant, and $\delta_2$ is inverse polylogarithmic.
\end{theorem}
Note that we can reduce the constant failure probability to inverse polynomial failure probability via standard techniques, at the cost
of increasing our space bound by a logarithmic factor. Observe that Problem~\ref{genprob} is a general case of Problem~\ref{motivating} where $g(x) = |x|$ (i.e., $L_1$ distance).
In the streaming model, we receive matrix $A$ implicitly, but we conceptualize the problem as if the matrix
were given explicitly and then resolve this issue by assuming we have blackbox algorithms that operate over the implicit matrix.

We define our stream such that each element in the stream $a_k$ is a pair of values $(i,j)$:

\begin{definition}[Stream]  Let $m,n$ be positive integers. A \textbf{stream} D = D(m,n) is a sequence of length m,
$a_1,a_2,\ldots,a_m$, where each entry is a pair of values in $\{0,\ldots,n\}$. 
\end{definition}

Let $g: \mathbb{R} \to \mathbb{R}$ be a non-negative, subadditive, and even function where $g(0) = 0$. Frequently,
we will need to discuss a matrix where $g$ has been applied to every entry. We use the notations
from~\cite{guillot2015complete} which are in turn based on notations from \cite{horntopics}.

\begin{definition}[Hadamard Function] Given Matrix $A$ of dimensions $n \times n $ a {\bf Hadamard function} $g$ takes as input
a matrix $A$ and is applied entrywise to every entry of the matrix. The output is matrix $g[A]$.
Further, we note that the $L_1$ norm of $g[A]$ is equivalent to the value we aim to approximate,
$\|g[A]\|_1 = \sum\limits_{i=1}^{n} \sum\limits_{j=1}^{n}g(a_{ij})$.
\end{definition}

We frequently use hash functions in our analysis, we now specify some notation. We sometimes express a hash function $H$ as a vector of values $ \{h_1,h_2,...,h_n\}$.
Multiplication of hash functions, denoted $H' = HAD(H^a,H^b)$, is performed entrywise such that  $\{h'_1 = h^a_1h^b_1 ,..., h'_n = h^a_nh^b_n\}$. 

We now define two additional matrices. All matrices in our definitions are of size $n \times n$, and all vectors are of size $1 \times n$.
We denote by $[n]$ the set $\{1,\ldots,n\}$.

\begin{definition}[Sampling Identity Matrix] Given a hash function $H:[n] \rightarrow \{0,1\}$, let
$h_i = H(i)$.  The {\bf Sampling Identity Matrix} $I_H$ with entries $b_{ij}$ is defined as:

\centerline{$I_H = \begin{cases} b_{ii} = h_i \\  b_{ij}  = 0 \mbox{ for } i \neq j. \end{cases}$}
\end{definition}

\noindent That is, the diagonal of $I_H$ are the values of $H$. When we multiply matrix $I_H$ by $A$, each row of
$I_H A$ is either the zero vector (corresponding to $h_i = 0$)
or the original row $i$ in matrix $A$ (corresponding to $h_i = 1$).  We use the term ``sampling" due to the fact that
the hash functions we use throughout this paper are random, and hence which rows remain untouched is random.
The same observations apply to columns when considering the matrix $A I_H$.

\begin{definition}[Row Aggregation Vector] A {\bf Row Aggregation Vector} $J$ is a $1 \times n$ vector with all entries equal to $1$.

\end{definition}

Thus, $JA$ yields a vector $V$ where each value $v_i$ is $\sum_{j=1}^n a_{ij}$.

\begin{ba}{ $(1\pm \epsilon')$-Approximation of $g$ on a Row Aggregate Vector} \label{BA1def}

Input: Matrix $A$, and hash function $H$. 

Output: $(1\pm \epsilon')$-approximation of  $\|g[JI_HA]\|_1$ with probability $(1 - \delta_1)$.

\end{ba}

\noindent The space Black Box Algorithm 1 (BA1) uses is referred to as $SPACE1(n,\delta_1,\epsilon')$ in our analysis.

\begin{ba} {$r(n)$-Approximation of $\|g[I_HA]\|_1$}\label{BA2def}

Input: Matrix $A$ and hash function $H$.

Output: An $r(n)$-approximation of $\|g[I_HA]\|_1$ with probability $(1 - \delta_2)$.

\end{ba}

\noindent The space Black Box Algorithm 2 ($BA2$) uses is referred to as $SPACE2(n,\delta_2)$ in our analysis.

\begin{definition}[Complement Hash Function]\label{hbar} For hash function $H:[n] \rightarrow \{0,1\}$ define the {\bf Complement Hash Function} $\bar H$ as $\bar{H}(i) = 1$ if and only if $H(i) = 0$.
\end{definition}

\begin{definition}[Threshold Functions] We define two threshold functions, which we denote by $\rho(n,\epsilon) = {r^4(n) \over \epsilon}$
and $\tau(n, \epsilon) = {2r^2(n) \over O(\epsilon)}$.
\end{definition}

\begin{definition}[Weight of a Row]\label{weightofrow}
The {\bf weight} of row $i$ in matrix $A$ is $u_{A,i}$, where $u_{A,i} = \sum\limits_{j=1}^n a_{ij}$.
\end{definition}

\begin{definition}[$\alpha$-Heavy Rows]\label{defn:heavy}
Row $i$ is {\bf $\alpha$-heavy} for $0 < \alpha < 1$ if $u_{A,i} > \alpha\|A\|_1$.
\end{definition}

\begin{definition}[Key Row]\label{defn:key}
We say row $i$ is a {\bf Key Row} if: $u_{A,i} > \rho(n,\epsilon)(\|A\|_1 - u_{A,i})$.
\end{definition}

\noindent While Definition~\ref{defn:heavy} and Definition~\ref{defn:key} are similar, we define them for convenience, as our algorithm
works by first finding key rows and then building on top of this to find $\alpha$-heavy rows.  We note that, as long as $\rho(n,\epsilon) \geq 1$,
a matrix can have at most one key row (since any matrix can have at most $\frac{1}{\alpha}$ $\alpha$-heavy rows, and a key row is $\alpha$-heavy
for $\alpha = \frac{\rho(n,\epsilon)}{1 + \rho(n,\epsilon)}$).

\section{Subadditive Approximations}\label{subadditive}

In this section we show that a $(1 \pm \epsilon)$-approximation of even,
subadditive, non-negative Hadamard functions which are zero at the
origin are preserved under row or column aggregations in the presence of
sufficiently heavy rows or columns.

\begin{theorem}\label{BmatrixProof}
Let $B$ be an $n \times n$ matrix and let $\epsilon\in (0,1)$ be a parameter. Recall that $J$ is a row vector with all entries equal to $1$.  Let
$g$ be any even, subadditive, non-negative Hadamard function which satisfies $g(0) = 0$. Denote
$u_i = \sum_{j=1}^n g(b_{ij})$, and
thus $\|g[B]\|_1 = \sum_{i=1}^n u_i$.
If there is a row $h$ such that $u_h \ge (1 - {\epsilon \over 2}) \|g[B]\|_1$ then
$
|u_h - \|g[JB]\|_1| \le \epsilon \|g[JB]\|_1.
$
\end{theorem}
\begin{proof}
Denote $V = JB$. Without loss of generality assume $u_1$ is the row such that $u_1 \ge (1 - {\epsilon \over 2}) \|g[B]\|_1$.
By subadditivity of $g$:
$ \|g[V]\|_1 \le \|g[B]\|_1 \le {1 \over  1 - {\epsilon \over 2}} u_1 \le (1 +{\epsilon }) u_1$.
Further, we have $b_{1j} = \left(\sum_{i=1}^n b_{ij} - \sum_{i=2}^n b_{ij} \right)$.
Since $g$ is even and subadditive, and such functions are non-negative, we have
$g(b_{1j}) \le g\left(  \sum_{i=1}^n b_{ij} \right) +  \sum_{i=2}^n g( b_{ij} )$.
Rearranging and summing over $j$, we get:
$\sum_{j=1}^n g \left( \sum_{i=1}^n  b_{ij} \right) \ge \sum_{j=1}^n \left( g(b_{1,j})  - \sum_{i=2}^n g( b_{ij}) \right)$.

Therefore:
\begin{equation*}
\|g[V]\|_1 = \sum_{j=1}^n g \left( \sum_{i=1}^n  b_{ij} \right) \ge \sum_{j=1}^n \left( g(b_{1,j}) - \left( \sum_{i=2}^n g( b_{ij}) \right) \right) = u_1 - (\|g[B]\|_1 - u_1).
\end{equation*}

Finally:

$$ \|g[V]\|_1 \ge u_1 - (\|g[B]\|_1 - u_1)  = 2u_1 - \|g[B]\|_1\ge u_1(2 - {1 \over  1 - {\epsilon \over 2}}) = u_1  {1 - \epsilon  \over  1 - {\epsilon \over 2}}  \ge u_1 (1- \epsilon).$$

\end{proof}

\section{Algorithm for Finding Key Rows}

\begin{definition}[Algorithm for Finding Key Rows]\leavevmode\label{HRalgdef}

Input: Matrix $A$ and Sampling Identity Matrix $I_H$ generated from hash function $H$.

Output: Pair (a, b), where the following holds for $a,b$, and the matrix $W = I_H A$:
\begin{enumerate}
\item The pair is either $(a,b) = (-1,0)$ or $(a,b) = (i,\tilde{u}_{W,i})$.  Here, $\tilde{u}_{W,i}$ is
a $(1 \pm \epsilon)$-approximation to $u_{W,i}$ and the index $i$ is the correct corresponding row.
\item If there is a key row $i_0$ for the matrix $W$, then $a = i_0$.
\end{enumerate}
\end{definition}

Before describing the algorithm and proving its correctness, we prove the following useful lemma in
Appendix~\ref{lemma1proof}.
\begin{lemma}\label{lem:balanced}
Let $U = (u_1,\ldots,u_n)$ be a vector with non-negative entries of dimension $n$ and let $H'$ be a pairwise independent
hash function where $H':[n] \to \{0,1\}$ and $P[H'(i) = 1] = P[H'(i) = 0] = \frac{1}{2}$.   Denote
by $\bar{H'}$ the hash function defined by $\bar{H'}(i) = 1 - H'(i)$.  Let $X = \sum_i H'(i)u_i$ and
$Y = \sum_i \bar{H'}(i)u_i$.  If there is no $\frac{1}{16}$-heavy element with respect to $U$, then:
$$\Pr\left[\left(X \leq \frac{1}{4} \cdot \|U\|_1\right) \cup \left(Y \leq \frac{1}{4} \cdot \|U\|_1\right)\right] \leq \frac{1}{4}.$$
\end{lemma}

\begin{theorem}{}\label{maintheorem}

If there exist two black box algorithms as specified in Black Box Algorithms~\ref{BA1def} and~\ref{BA2def},
then there exists an algorithm that satisfies the requirements in Definition~\ref{HRalgdef} with high probability.

\end{theorem}

\begin{algorithm}[ht]
\label{KeyRowAlg} 
\caption{Algorithm Find-Key-Row}

The algorithm takes as input matrix $A$ and hash function $H:[n] \rightarrow \{0,1\}$

\begin{algorithmic}

\For{$\ell$ = 1 to $N = O(\log n)$}

\State Generate a pairwise independent, uniform hash function $H_\ell: [n] \to \{0,1\}$

\State Let $T_{1} = HAD(H,H_{\ell})$, $T_{0} = HAD(H,\bar{H}_{\ell})$

\State Let $y_1 = BA2(A,T_1)$, $y_0 = BA2(A,T_0)$ ($BA2$ is run with constant failure probability $\delta_2$)

\If{$y_0 \ge \tau(n,\epsilon) \cdot y_1$}

\State $b_\ell = 0$

\ElsIf{$y_1 \ge \tau(n,\epsilon) \cdot y_0$}

\State $b_\ell = 1$

\Else

\State $b_\ell = 2$

\EndIf
\EndFor

\If{$ |\{ \ell : b_\ell = 2 \}| \ge \frac{2}{3} n$}

\State Return $(-1,0)$

\Else

	\If{there is a row $i$ such that $i$ satisfies $|\{\ell : H_{\ell}(i) = b_\ell\}| \geq \frac{3}{4} \cdot N$ }

	\State Return $(i,BA1(A,H))$ ($BA1$ is run with $\epsilon' = \frac{\epsilon}{2}$ and $\delta_1$ is set to be inverse polylogarithmic)

	\Else 
	
	\State Return $(-1,0)$

	\EndIf

\EndIf

\end{algorithmic}
\label{alg:key}
\end{algorithm}

\begin{proof}

We will prove that Algorithm~\ref{alg:key} fits the description of Definition~\ref{HRalgdef}.  Using standard methods such as
in~\cite{Braverman:2010:ZFL:1806689.1806729}, we have a loop that runs in parallel $O(\log(n))$ times so that we can
find the index of a heavy element and return it, if there is one.  To prove this theorem, we consider the following three
exhaustive and disjoint cases regarding the matrix $g[I_H A]$ (recall that $H : [n] \rightarrow \{0,1\}$):
\begin{enumerate}
\item The matrix has a key row (note that a matrix always has at most one key row).

\item The matrix has no $\alpha$-heavy row for $\alpha = 1 - \frac{\epsilon}{8}$.

\item The matrix has an $\alpha$-heavy row for $\alpha = 1 - \frac{\epsilon}{8}$, but there is no key row.
\end{enumerate}

We prove that the algorithm is correct in each case in Lemmas~\ref{lem:case1},~\ref{lem:case2}, and~\ref{lem:case3}, respectively.
These proofs can be found in Appendix~\ref{KeyCases}.
\end{proof}

With the proof of these three cases, we are done proving that Algorithm~\ref{alg:key} performs correctly. We now analyze the space bound for Algorithm~\ref{alg:key}.

\begin{lemma}\label{HHAspace}
Algorithm~\ref{alg:key} uses $O\left(SPACE1(n,\delta_1,\frac{\epsilon}{2}) + \log(n)(\log^2(n) + SPACE2(n,\delta_2))\right)$
bits of memory, where $\delta_1$ is inverse polylogarithmic and $\delta_2$ is a constant.
\end{lemma}
\begin{proof}
Note that, in order for our algorithm to succeed, we run $BA1$ with an error parameter of
$\epsilon' = \frac{\epsilon}{2}$ and a failure probability
parameter $\delta_1$ which is inverse polylogarithmic.  Moreover, we run $BA2$ with a constant
failure probability.  We also require a number of random bits bounded by $O(\log^2(n))$ for generating each hash function $H_{\ell}$,
as well as the space required to run $BA2$ in each iteration of the loop.  Since there are $O(\log n)$ parallel iterations,
this gives the lemma.
\end{proof}

\subsection{Algorithm for Finding All $\alpha$-Heavy Rows}\label{anyHH}

Algorithm~\ref{alg:key} only guarantees that we return key rows.  Given a matrix $A$, we now show that this algorithm can be
used as a subroutine to find all $\alpha$-heavy rows $i$ with respect to the matrix $g[A]$ with high probability, along with
a $(1\pm \epsilon)$-approximation to the row weights $u_{g[A],i}$ for all $i$. In order to do this, we apply an additional
hash function $H:[n] \to [\tau]$ which essentially maps rows of the matrix
to some number of buckets $\tau$ (i.e., each bucket corresponds to a set of sampled rows based on $H$), and then run
Algorithm~\ref{alg:key} for each bucket.  The intuition for why the algorithm works is that any $\alpha$-heavy row $i$ in the original
matrix $A$ is likely to be a key row for the matrix in the corresponding bucket to which row $i$ is mapped.
Note that, eventually, we find $\alpha$-heavy rows for $\alpha = \frac{\epsilon^2}{\log^3 n}$.
The algorithm sets $\tau = O\left(\frac{\rho(n,\epsilon)\log(n)}{\alpha^2}\right)$ and is given below.

\begin{algorithm}[H]
\caption{Algorithm Find-Heavy-Rows}

This algorithm takes as input a matrix $A$ and a value $0 < \alpha < 1$.

\begin{algorithmic}

\State Generate a pairwise independent hash function $H: [n] \to [\tau]$, where $\tau = O\left(\frac{\rho(n,\epsilon)\log(n)}{\alpha^2}\right)$

\For{$k = 1$ to $\tau$}
\State Let $H_k: [n] \to \{0,1\}$ be the function defined by $H_k(i) = 1 \Longleftrightarrow H(i) = k$

\State Let $C_k = \textrm{Find-Key-Row}(A,H_k)$
\EndFor
\State Return $\{C_k : C_k \neq (-1,0) \}$
\end{algorithmic}
\label{alg:heavy}
\end{algorithm}

\begin{theorem}\label{keyheavythm}
Algorithm~\ref{alg:heavy} outputs a set of pairs $Q = \{(i_1,a_1),\ldots,(i_t,a_t)\}$ for $t \leq \tau$
which satisfy the following properties, except with probability $\frac{1}{\log n}$:

\begin{enumerate}
\item $\forall j\in [t]: (1-\epsilon)u_{g[A],i_j} \leq a_{j} \leq (1 \pm \epsilon)u_{g[A],i_j}$.
\item $\forall i\in [n]$: If row $i$ is $\alpha$-heavy with respect to the matrix $g[A]$, then $\exists j\in [t]$ such that $i_j = i$ (for any $0 < \alpha < 1$).
\end{enumerate}
\end{theorem}
\begin{proof}
First, the number of pairs output by Algorithm~\ref{alg:heavy} is at most the number of buckets, which equals $\tau$.
Now, the first property is true due to the fact that Algorithm~\ref{alg:key} has a high success probability.  In particular,
as long as the failure probability is at most $\frac{1}{\tau \cdot \log^c(n)}$ for some constant $c$ (which we ensure), then by union bound
the probability that there exists a pair $(i_j,a_j) \in Q$ such that $a_j$ is not a $(1\pm\epsilon)$-approximation to
$u_{g[A],i_j}$ is at most inverse polylogarithmic.

Now, to ensure the second item, we need to argue that every $\alpha$-heavy row gets mapped to its own bucket with high probability, since if there is
a collision the algorithm cannot find all $\alpha$-heavy rows.  Moreover, we must argue that
for each $\alpha$-heavy row $i$ with respect to the matrix $g[A]$, if $i$ is mapped to bucket $k$ by $H$, then row $i$ is actually a key row in the
corresponding sampled matrix $g[A_k]$ (for ease of notation, we write $A_k$ to denote the matrix $H_kA_k$).
More formally, suppose row $i$ is $\alpha$-heavy.  Then the algorithm must guarantee
with high probability that, if $H(i) = k$, then row $i$ is a key row in the matrix $g[A_k]$.  If we prove these two properties,
then the theorem holds (since Algorithm~\ref{alg:key} outputs a key row with high probability, if there is one).

Observe that there must be at most $\frac{1}{\alpha}$ rows which are $\alpha$-heavy.  In particular, let $R$ be the set of $\alpha$
heavy rows, and assume towards a contradiction that $|R| > \frac{1}{\alpha}$.  Then we have:
$$ \|g[A]\|_1 \geq \sum_{i \in R}u_{g[A],i} \geq \sum_{i \in R}\alpha \|g[A]\|_1 = \alpha \cdot \|g[A]\|_1 \cdot |R| > \|g[A]\|_1,$$
which is a contradiction.  Hence, we seek to upper bound the probability of a collision when throwing $\frac{1}{\alpha}$ balls into
$\tau$ bins.  By a Birthday paradox argument, this happens with probability at most $\frac{1}{2\cdot\tau\cdot\alpha^2}$, which
can be upper bounded as follows:
$$ \frac{1}{2 \tau \alpha^2} \leq \frac{\alpha^2}{2\alpha^2 \rho(n,\epsilon) \log(n)} = \frac{1}{2\rho(n,\epsilon)\log(n)} = \frac{\epsilon}{2r^4(n)\log(n)},$$
which is inverse polylogarithmically small.

Now, we argue that every $\alpha$-heavy row $i$ for the matrix $g[A]$ is mapped to a sampled matrix such that $i$ is a key row in the
sampled matrix with high probability.  In particular, suppose $H(i) = k$, implying that row $i$ is mapped to bucket $k$.
For $\ell \neq i$, let $X_{\ell}$ be the indicator random variable which is $1$ if and only if row $\ell$ is mapped to the same bucket as $i$,
namely $H(\ell) = k$ (i.e., $X_\ell = 1$ means the sampled matrix $g[A_k]$ contains row $i$ and row $\ell$).
If row $i$ is not a key row for the matrix $g[A_k]$, this means that $u_{g[A_k],i} \leq \rho(n,\epsilon)(\|g[A_k]\|_1 - u_{g[A_k],i})$.
Observe that, if row $i$ is mapped to bucket $k$, then we have $u_{g[A_k],i} = u_{g[A],i}$.  Hence, the
the probability that row $i$ is not a key row for the sampled matrix $g[A_k]$ (assuming row $i$ is mapped to bucket $k$)
can be expressed as $\Pr[u_{g[A],i} \leq \rho(n,\epsilon)(\|g[A_k]\|_1 - u_{g[A],i}) | H(i) = k]$.  By pairwise independence
of $H$, and by Markov's inequality, we can write:
\begin{align*}
\Pr\Big[u_{g[A],i} &\leq \rho(n,\epsilon)(\|g[A_k]\|_1 - u_{g[A],i}) \ \Big| \ H(i) = k\Big] \\
&= \Pr\left[u_{g[A],i} \leq \rho(n,\epsilon)\sum_{\ell \neq i}u_{g[A],\ell}X_\ell \ \middle| \ H(i) = k\right] =
\Pr\left[u_{g[A],i} \leq \rho(n,\epsilon)\sum_{\ell \neq i}u_{g[A],\ell}X_\ell \right] \\
&= \Pr\left[\sum_{\ell \neq i}u_{g[A],\ell}X_\ell \geq \frac{u_{g[A],i}}{\rho(n,\epsilon)}\right] \leq
\frac{\rho(n,\epsilon)\mathbb{E}\left[\sum_{\ell \neq i}u_{g[A],\ell}X_\ell\right]}{u_{g[A],i}} = \frac{\rho(n,\epsilon)\sum_{\ell \neq i}u_{g[A],\ell}}{\tau \cdot u_{g[A],i}}\\
&\leq \frac{\rho(n,\epsilon) \|g[A]\|_1}{\alpha \tau \|g[A]\|_1} = \frac{\alpha^2 \rho(n,\epsilon)}{4 \alpha \cdot \rho(n,\epsilon) \log(n)} \leq \frac{\alpha}{4\log(n)}.
\end{align*}
Here, we choose $\tau = \frac{4\rho(n,\epsilon)\log(n)}{\alpha^2}$, and get that the probability that a particular $\alpha$-heavy row $i$ is not a key
row in its corresponding sampled matrix is at most $\frac{\alpha}{4\log(n)}$.  Since there are at most $\frac{1}{\alpha}$ rows which are $\alpha$-heavy,
by union bound the probability that there exists an $\alpha$-heavy row that is not a key row in its sampled matrix is at most $\frac{1}{4\log(n)}$.

Thus, in all, the probability that at least one bad event happens (i.e., there exists a pair $(i_j,a_j)$ such that $a_j$ is not a good approximation to
$u_{g[A],i_j}$, there is a collision between $\alpha$-heavy rows, or an $\alpha$-heavy row is not a key row in its corresponding sampled matrix)
is at most $\frac{1}{\log(n)}$.  This gives the theorem.
\end{proof}

\subsection{Sum from $\alpha$-Heavy Rows}\label{sec:sumfromheavy}
We now have an algorithm that is able to find all $\alpha$-heavy rows for $\alpha = \frac{\epsilon^2}{\log^3 n}$, except with
probability $\frac{1}{\log n}$.  In the language of~\cite{recursive1}, by Theorem~\ref{keyheavythm}, our $\alpha$-heavy rows algorithm outputs
an $(\alpha,\epsilon)$-cover with respect to the vector $(u_{g[A],1},u_{g[A],2},\ldots,u_{g[A],n})$ except with probability $\frac{1}{\log n}$,
where $\epsilon > 0$ and $\alpha > 0$.  Hence, we can apply the Recursive Sum algorithm from~\cite{recursive1} (see Appendix~\ref{recursive}
for the formal definition of an $(\alpha,\epsilon)$-cover, along with the Recursive Sum algorithm) to get a $(1 \pm \epsilon)$-approximation
of $\|g[A]\|_1$.  Note that the Recursive Sum algorithm needs $\alpha = \frac{\epsilon^2}{\log^3 n}$ and a failure probability of at most
$\frac{1}{\log n}$, which we provide.  Hence, we get the following theorem.
\begin{theorem}\label{recthem}
The Recursive Sum Algorithm, using Algorithm \ref{alg:heavy} as a subroutine, returns a $(1 \pm \epsilon)$-approximation of $\|g[A]\|_1$.  
\end{theorem}

\subsection{Space Bounds}

\begin{lemma}
Recursive Sum, using Algorithm~\ref{alg:heavy} as a subroutine as described in Section~\ref{sec:sumfromheavy}, uses the following
amount of memory, where $\epsilon' = \frac{\epsilon}{2}$, $\delta_1$ is inverse polylogarithmic, and $\delta_2$ is a small constant:
$$O\left( {1 \over \epsilon^5} r^4(n)\log^{10}(n) + {1 \over \epsilon^5} r^4(n)\log^8(n)SPACE1(n,\delta_1,\epsilon') +  {1 \over \epsilon^5} r^4(n)\log^9(n)SPACE 2(n,\delta_2)\right).$$
\end{lemma}

\begin{proof}
The final algorithm uses the space bound from Lemma \ref{HHAspace}, multiplied by
$\tau = O\left(\frac{\rho(n,\epsilon)\log(n)}{\alpha^2}\right)$, where $\alpha = {\epsilon^2 \over \phi^3}$,
$\phi = O(\log n)$, and $\rho(n,\epsilon) =  {r^4(n) \over \epsilon}$. This gives $\tau = {1 \over \epsilon^5}r^4(n)\log^7(n)$
to account for the splitting required to find $\alpha$-heavy rows in Section \ref{anyHH}. Finally, a multiplicative
cost of  $\log(n)$ is needed for Recursive Sum, giving the final bound.
\end{proof}

\section{Applications}

We now apply our algorithm to the problem of determining the $L_1$ distance between joint and product distributions as described in Problem \ref{motivating}.

\subsection*{Space Bounds for Determining $L_1$ Independence}

Given an $n \times n$ matrix $A$ with entries $a_{ij} = g\left( {f_{ij} \over m}  - {f_{i}f_{j} \over m}  \right)$, we have
provided a method to approximate:
$$  \sum\limits_{i=1}^n \sum\limits_{j = 1}^n g\left( {f_{ij} \over m}  - {f_{i}f_{j} \over m}  \right).$$

Let $g$ be the $L_1$ distance, namely $g(x) = |x|$.  We now state explicitly which blackbox algorithms we use:

\begin{itemize}
\item Let Black Box Algorithm 1 ($BA1$) be the $(1 \pm \epsilon)$-approximation of $L_1$  for vectors from \cite{stable}.
The space of this algorithm is upper bounded by the number of random bits required and uses
$O(\log({nm \over \delta \epsilon}) \log({ m \over \delta \epsilon} ) \log ( {1 \over \delta} ) \epsilon^{-2}  )$ bits of memory.

\item Let Black Box Algorithm 2 ($BA2$) be the $r(n)$-approximation, using the $L_1$ sketch of the distance
between joint and product distributions from \cite{IM08}. This algorithm does not have a precise polylogarithmic
bound provided, but we compute that it is upper bounded by the random bits required to generate the Cauchy random variables similarly to $BA1$.
This algorithm requires $O(\log({nm \over \delta \epsilon}) \log({ m \over \delta \epsilon} ) \log ( {1 \over \delta} ) \epsilon^{-2} )$
bits of memory.

\end{itemize}
These two algorithms match the definitions given in Section \ref{SFandBA}, thus we are able to give a bound
of $O({ 1 \over \epsilon^7}\log^{14}(n) \log^{2}({nm \over \epsilon}))$ on the space our algorithm requires.
We can improve this slightly as follows.
\begin{corollary}Due to the nature of the truncated Cauchy distribution (see 
\cite{IM08}), we can further improve our space bound to $O\left({ 1 \over \epsilon^7}\log^{12}(n) \log^{2}({nm \over \epsilon})\right)$.

\begin{proof}

Due to the constant lower bound on the approximation of $L_1$, instead of ${1 \over r^2(n)} \le \|g[W]\|_1 \le r^2(n)$, we get $C \le \|g[W]\|_1 \le log^2(n)$
for some constant $C$.
As the space cost from dividing the matrix into submatrices as shown in Section \ref{anyHH} directly depends on these bounds, we only pay an $O(r^2(n))$
multiplicative factor instead of an $O(r^4(n))$ multiplicative factor and achieve a bound of $O\left({ 1 \over \epsilon^7}\log^{12}(n) \log^{2}({nm \over \epsilon})\right)$.
\end{proof}

\end{corollary}

\newpage

\bibliography{Bibliography}

\begin{thebibliography}{10}

\bibitem{alon2007testing}
Noga Alon, Alexandr Andoni, Tali Kaufman, Kevin Matulef, Ronitt Rubinfeld, and
  Ning Xie.
\newblock Testing k-wise and almost k-wise independence.
\newblock In {\em Proceedings of the 39th annual ACM Symposium on Theory of
  Computing}, 2007.

\bibitem{alon2003almost}
Noga Alon, Oded Goldreich, and Yishay Mansour.
\newblock Almost k-wise independence versus k-wise independence.
\newblock {\em Information Processing Letters}, 88(3):107--110, 2003.

\bibitem{ams}
Noga Alon, Yossi Matias, and Mario Szegedy.
\newblock The space complexity of approximating the frequency moments.
\newblock In {\em Proceedings of the 28th annual ACM Symposium on Theory of
  Computing}, 1996.

\bibitem{BatuTesting2001}
Tu{\u{g}}kan Batu, Eldar Fischer, Lance Fortnow, Ravi Kumar, Ronitt Rubinfeld,
  and Patrick White.
\newblock Testing random variables for independence and identity.
\newblock In {\em Proceedings of the 42nd annual IEEE Symposium on Foundations
  of Computer Science}, 2001.

\bibitem{BatuIndependence}
Tu{\u{g}}kan Batu, Lance Fortnow, Ronitt Rubinfeld, Warren~D. Smith, and
  Patrick White.
\newblock Testing that distributions are close.
\newblock In {\em Proceedings of the 41st annual IEEE Symposium on Foundations
  of Computer Science}, 2000.

\bibitem{batu2013testing}
Tu{\u{g}}kan Batu, Lance Fortnow, Ronitt Rubinfeld, Warren~D. Smith, and
  Patrick White.
\newblock Testing closeness of discrete distributions.
\newblock {\em Journal of the ACM}, 60(1):4, 2013.

\bibitem{batu2004sublinear}
Tu{\u{g}}kan Batu, Ravi Kumar, and Ronitt Rubinfeld.
\newblock Sublinear algorithms for testing monotone and unimodal distributions.
\newblock In {\em Proceedings of the 36th annual ACM Symposium on Theory of
  Computing}, 2004.

\bibitem{braverman2010ams4}
Vladimir Braverman, Kai-Min Chung, Zhenming Liu, Michael Mitzenmacher, and
  Rafail Ostrovsky.
\newblock {AMS Without 4-Wise Independence on Product Domains}.
\newblock In {\em {Proceedings of the 27th International Symposium on
  Theoretical Aspects of Computer Science}}, 2010.

\bibitem{BO10}
Vladimir Braverman and Rafail Ostrovsky.
\newblock Measuring independence of datasets.
\newblock In {\em Proceedings of the 42nd annual ACM Symposium on Theory of
  Computing}, 2010.

\bibitem{Braverman:2010:ZFL:1806689.1806729}
Vladimir Braverman and Rafail Ostrovsky.
\newblock Zero-one frequency laws.
\newblock In {\em Proceedings of the 42nd ACM Symposium on Theory of
  Computing}, 2010.

\bibitem{recursive1}
Vladimir Braverman and Rafail Ostrovsky.
\newblock Generalizing the layering method of {I}ndyk and {W}oodruff: Recursive
  sketches for frequency-based vectors on streams.
\newblock In {\em International Workshop on Approximation Algorithms for
  Combinatorial Optimization Problems}. Springer, 2013.

\bibitem{guillot2015complete}
Dominique Guillot, Apoorva Khare, and Bala Rajaratnam.
\newblock Complete characterization of hadamard powers preserving loewner
  positivity, monotonicity, and convexity.
\newblock {\em Journal of Mathematical Analysis and Applications},
  425(1):489--507, 2015.

\bibitem{horntopics}
Roger~A. Horn and Charles~R. Johnson.
\newblock Topics in matrix analysis.
\newblock {\em Cambridge University Presss, Cambridge}, 1991.

\bibitem{stable}
Piotr Indyk.
\newblock Stable distributions, pseudorandom generators, embeddings, and data
  stream computation.
\newblock {\em Journal of the ACM}, 53(3):307--323, 2006.

\bibitem{IM08}
Piotr Indyk and Andrew McGregor.
\newblock Declaring independence via the sketching of sketches.
\newblock In {\em Proceedings of the 19th annual ACM-SIAM Symposium on Discrete
  Algorithms}, 2008.

\bibitem{kimball2004data}
Ralph Kimball and Joe Caserta.
\newblock The data warehouse etl toolkit: Practical techniques for extracting,
  cleaning, conforming, and delivering data.
\newblock 2004.

\bibitem{lehmann2006testing}
Erich~L. Lehmann and Joseph~P. Romano.
\newblock {\em Testing statistical hypotheses}.
\newblock Springer Science \& Business Media, 2006.

\bibitem{mcgregorevaluating}
Andrew McGregor and Hoa~T. Vu.
\newblock Evaluating bayesian networks via data streams.
\newblock In {\em Proceedings of the 21st International Computing and
  Combinatorics Conference}, 2015.

\bibitem{poosala1997selectivity}
Viswanath Poosala and Yannis~E. Ioannidis.
\newblock Selectivity estimation without the attribute value independence
  assumption.
\newblock In {\em Proceedings of the 23rd International Conference on Very
  Large Data Bases}, 1997.

\bibitem{Rubinfeld05testingmonotone}
Ronitt Rubinfeld and Rocco~A. Servedio.
\newblock Testing monotone high-dimensional distributions.
\newblock In {\em Proceedings of the 37th annual ACM Symposium on Theory of
  Computing}, 2005.

\bibitem{Sahai2003Complete}
Amit Sahai and Salil Vadhan.
\newblock A complete problem for statistical zero knowledge.
\newblock {\em Journal of the ACM}, 50(2):196--249, 2003.

\bibitem{W14}
David~P. Woodruff.
\newblock Sketching as a tool for numerical linear algebra.
\newblock {\em CoRR}, abs/1411.4357, 2014.

\end{thebibliography}

\appendix

\section{Proof of of Lemma 1}\label{lemma1proof}

\begin{proof}
Note that we always have the equality $X+Y = \sum_i H'(i)u_i + \bar{H'}(i)u_i = \sum_i H'(i)u_i + (1-H'(i))u_i = \|U\|_1$,
and moreover $\mathbb{E}[X] = \sum_i u_i \mathbb{E}[H'(i)] = \frac{1}{2} \cdot \|U\|_1$.  Also, observe that
\begin{align*}
Var[X] = \mathbb{E}[X^2] - (\mathbb{E}[X])^2 &= \sum_i \mathbb{E}[(H'(i))^2]u^2_i + \sum_{i \neq j}\mathbb{E}[H'(i)H'(j)]u_iu_j - \frac{1}{4} \cdot \|U\|^2_1 \\
&= \frac{1}{2}\sum_i u^2_i + \frac{1}{4}\sum_{i \neq j}u_i u_j - \frac{1}{4}\left(\sum_i u^2_i + \sum_{i \neq j}u_i u_j\right) = \frac{1}{4}\sum_iu^2_i.
\end{align*}
Using the fact that there is no $\frac{1}{16}$-heavy element with respect to $U$, which implies that
$u_i \leq \frac{1}{16}\cdot\|U\|_1$ for all $i$, we have:
$$Var[X] = \frac{1}{4}\sum_i u^2_i \leq \frac{\|U\|_1}{64} \sum_iu_i = \frac{\|U\|^2_1}{64}.$$
Now we can apply Chebyshev's inequality to obtain:
\begin{align*}
\Pr\left[\left(X \leq \frac{1}{4} \cdot \|U\|_1\right) \cup \left(Y \leq \frac{1}{4} \cdot \|U\|_1\right)\right] &= \Pr\left[|X - \mathbb{E}[X]| \geq \frac{\|U\|_1}{4}\right]\\
&\leq \frac{16\cdot Var[X]}{\|U\|^2_1} \leq \frac{16 \cdot \|U\|^2_1}{64\cdot \|U\|^2_1} = \frac{1}{4}.
\end{align*}
\end{proof}

\section{Proof of Correctness of Algorithm 1}\label{KeyCases}

Throughout the lemmas, we imagine that the hash function $H:[n] \rightarrow \{0,1\}$ is fixed, and hence the matrix $g[I_H A]$
is fixed.  All randomness is taken over the pairwise independent hash functions $H_\ell$ that are generated
in parallel, along with both blackbox algorithms.

To ease the notation, we define
$$W = I_H A \textrm{, } W_1 = I_{T_1}A \textrm{, and } W_0 = I_{T_2}A$$
(recall the notation from Algorithm~\ref{alg:key} that $T_{1} = HAD(H,H_{\ell})$ and $T_{0} = HAD(H,\bar{H}_{\ell})$).
Finally, for each row
$i$ in the matrix $g[W]$, we define the shorthand notation $u_i = u_{g[W],i}$.

\begin{lemma}\label{lem:case1}
If the matrix $g[I_H A]$ has a key row, Algorithm~\ref{alg:key} correctly returns the index of the row and a
$(1 \pm \epsilon)$-approximation of $\|g[I_H A]\|_1$ except with inverse polylogarithmic probability.
\end{lemma}

\begin{proof}

Suppose the matrix $g[I_H A]$ has a key row, and let $i_0$ be the index of this row.
We prove that we return a good approximation of $u_{g[W],i_0}$ with high probability.  In particular,
we first argue that, for a fixed iteration $\ell$ of the loop, we have the property that $b_\ell$ equals
$H_{\ell}(i_0)$, and moreover this holds with certainty.  We assume without loss of generality that $H_{\ell}(i_0) = 1$
(the case when $H_{\ell}(i_0) = 0$ is symmetric).  In particular, this implies that the key row $i_0$ appears
in the matrix~$g[W_1]$.

By definition of $BA2$, the following holds for $y_1 = BA2(A,T_1)$ and $y_0 = BA2(A,T_0)$,
except with probability $2\delta_2$ (where $\delta_2$ is the failure probability of $BA2$):
$$y_1 \ge  \frac{\|g[W_1]\|_1}{r(n)} \textrm{ and } y_0 \le \|g[W_0]\|_1 r(n).$$

We have the following set of inequalities:
\begin{equation*} \label{boundSH}
\|g[W_1]\|_1 \geq u_{i_0} > \rho(n,\epsilon)(\|g[W]\|_1 - u_{i_0}) \geq \rho(n,\epsilon)\|g[W_0]\|_1,
\end{equation*}
where the first inequality follows since $g$ is non-negative and the key row $i_0$ appears in the matrix $g[W_1]$ (and hence the $L_1$-norm
of $g[W_1]$ is at least $u_{i_0}$ since it includes the row $i_0$), the second inequality follows by definition of $i_0$ being
a key row for the matrix $W$, and the last inequality follows since the entries in row $i_0$ of the matrix
$W_0$ are all zero (as $H_{\ell}(i_0) = 1$) and the remaining rows of $W_0$ are sampled from $W$, along with the
facts that $g$ is non-negative and $g(0) = 0$.

Substituting for  $\rho(n,\epsilon)$, and using the fact that $y_1$ and $y_0$ are good approximations for
$\|g[W_1]\|_1$ and $\|g[W_0]\|_1$ (respectively), except with probability $2\delta_2$, we get:
\begin{equation*}\label{case1bound}
y_1 \geq \frac{\|g[W_1]\|_1}{r(n)} > \frac{r^3(n)}{\epsilon} \cdot \|g[W_0]\|_1 \geq \frac{r^2(n)}{\epsilon} \cdot y_0 \geq \tau(n,\epsilon) \cdot y_0,
\end{equation*}
and thus in this iteration of the loop we have $b_\ell = 1$ except with probability $2\delta_2$
(in the case that $H_{\ell}(i_0) = 0$, it is easy to verify by
a similar argument that $y_0 \geq \tau(n,\epsilon) \cdot y_1$, and hence we have $b_\ell = 0$).  Hence, for the row $i_0$,
we have the property that $b_\ell = H_{\ell}(i_0)$ for a fixed $\ell$, except with probability $2\delta_2$.
By the Chernoff bound, as long as $\delta_2$ is a sufficiently small constant, we have $b_{\ell} = H_{\ell}(i_0)$
for at least a $\frac{3}{4}$-fraction of iterations $\ell$, except with inverse polynomial probability.  The only issue to consider is the case that there
exists another row $i \neq i_0$ with the same property, namely $b_\ell = H_{\ell}(i)$ for a large fraction of iterations $\ell$.
However, if $b_{\ell} = H_{\ell}(i)$, it must be that at least one of $y_1,y_0$ is a bad approximation or $H_{\ell}(i) = H_{\ell}(i_0)$,
which happens with probability at most $2\delta_2 + \frac{1}{2}$.  Therefore, by the Chernoff bound, the probability that this happens for
at least a $\frac{3}{4}$-fraction of iterations $\ell$ is at most $\frac{1}{2^{O(\log n)}}$, which is inverse polynomially small.
By applying the union bound, the probability that there exists such a row is at most $\frac{n-1}{2^{O(\log n)}}$, which is at most
an inverse polynomial.  Hence, in this case, the algorithm returns $(i_0,BA1(A,H))$ except with inverse polynomial
probability.

We now argue that $\tilde{u}_{g[W],i_o} = BA1(A,H)$ is a $(1\pm \epsilon)$-approximation of $u_{g[W],i_0}$, except with inverse
polylogarithmic probability.  By definition of $BA1$,
which we run with an error parameter of $\epsilon' = \frac{\epsilon}{2}$, it returns a $\left(1 \pm \frac{\epsilon}{2}\right)$-approximation of
$\|g[JW]\|_1$ except with inverse polylogarithmic probability, where $W = I_H A$.  Moreover, since $i_0$ is a key row, we have:
$$ u_{i_0} > \rho(n,\epsilon)(\|g[W]\|_1 - u_{i_0}) \Rightarrow u_{i_0} > \frac{\rho(n,\epsilon)\|g[W]\|_1}{1 + \rho(n,\epsilon)}
\geq \left(1 - \frac{\epsilon}{8}\right)\|g[W]\|_1, $$
where the last inequality follows as long as $r^4(n) \geq 8 - \epsilon$.  This implies that $i_0$ is $\left(1 - \frac{\epsilon}{8}\right)$-heavy
with respect to the matrix $g[W]$, and hence we can apply Theorem~\ref{BmatrixProof} to get that:
\begin{align*}
(1 \pm \epsilon)u_{i_0} \geq \frac{\left(1 + \frac{\epsilon}{2}\right)}{\left(1 - \frac{\epsilon}{4}\right)}u_{i_0}
&\geq \left(1 + \frac{\epsilon}{2}\right)\|g[JW]\|_1 \geq \tilde{u}_{g[W],i_0}\\
&\geq \left(1 - \frac{\epsilon}{2}\right)\|g[JW]\|_1 \geq
\frac{\left(1 - \frac{\epsilon}{2}\right)}{\left(1+\frac{\epsilon}{4}\right)}u_{i_0} \geq (1-\epsilon)u_{i_0},
\end{align*}
where the first inequality holds for any $0 < \epsilon \leq 1$, the second inequality holds by Theorem~\ref{BmatrixProof},
the third inequality holds since $\tilde{u}_{g[W],i_0}$ is a $\left(1 \pm \frac{\epsilon}{2}\right)$-approximation of $\|g[JW]\|_1$,
and the rest hold for similar reasons.  Hence, our algorithm returns a good approximation as long as $BA1$ succeeds.  Noting that
this happens except with inverse polylogarithmic probability gives the lemma.
\end{proof}

\begin{lemma}\label{lem:case2} If the input matrix has no $\alpha$-heavy row, where $\alpha = 1 - \frac{\epsilon}{8}$,
then with high probability Algorithm~\ref{alg:key} correctly returns $(-1,0)$.
\end{lemma}

\begin{proof}

In this case, we have no $\alpha$-heavy row for $\alpha = 1 - \frac{\epsilon}{8}$, which implies
that $u_i \leq \alpha \|g[W]\|_1 = \left(1-\frac{\epsilon}{8}\right)\|g[W]\|_1$ for each row $i$ in the matrix $g[W]$.
In this case, we show the probability that Algorithm~\ref{alg:key} returns a false positive is small.
That is, with high probability, in each iteration $\ell$ of the loop the algorithm sets $b_\ell = 2$, and hence it
returns $(-1,0)$.  We split this case into three additional disjoint and exhaustive subcases, defined as follows:
\begin{enumerate}
\item For each row $i$, we have $u_i \leq \frac{1}{16} \|g[W]\|_1$.
\item There exists a row $i$ with $u_i > \frac{1}{16} \|g[W]\|_1$ and $\forall j \neq i$ we have $u_j \leq \frac{\epsilon}{128}u_i$.
\item There exist two distinct rows $i,j$ where $u_i > \frac{1}{16} \|g[W]\|_1$ and $u_j > \frac{\epsilon}{128}u_i$.
\end{enumerate}
We define $X = \sum_i h_i^{\ell} u_i$ and $Y = \sum_i \bar{h}_i^{\ell} u_i$, where $h_i^{\ell} = H_{\ell}(i)$ and $\bar{h}_i^{\ell} = \bar{H}_{\ell}(i)$. 
Hence, we have $X = \|g[W_1]\|_1$ and $Y = \|g[W_0]\|_1$, and moreover $X + Y = \|g[W]\|_1$ (recall that $g[W_1] = g[I_{T_1}A]$ and
$g[W_0] = g[I_{T_0}A]$). 

In the first subcase, where there is no $\frac{1}{16}$-heavy row, we can apply Lemma~\ref{lem:balanced} to the vector
$(u_1,\ldots,u_n)$ to get that:

$$\Pr\left[\left(X \leq \frac{\|g[W]\|_1}{4} \right) \cup \left(Y \leq \frac{\|g[W]\|_1}{4} \right) \right] \leq \frac{1}{4}.$$

By definition of $BA2$, the following holds for $y_1 = BA2(A,T_1)$ and $y_0 = BA2(A,T_0)$ except with probability $2\delta_2$,
where $\delta_2$ is the success probability of $BA2$:
$$\frac{\|g[W_1]\|_1}{r(n)} \leq y_1 \leq r(n)\|g[W_1]\|_1 \textrm{ , } \quad \frac{\|g[W_0]\|_1}{r(n)} \leq y_0 \leq  r(n)\|g[W_0]\|_1.$$

Hence, except with probability $\frac{1}{4} + 2\delta_2$, we have the following constraints on $y_0$ and $y_1$:
\begin{align*}
y_0 \leq r(n)Y \leq r(n) \cdot \frac{3}{4} \cdot \|g[W]\|_1 &\leq 3 r(n) X \leq 3 y_1 r^2(n) \leq \tau(n,\epsilon) \cdot y_1, \textrm{ and}\\
y_1 \leq r(n)X \leq r(n) \cdot \frac{3}{4} \cdot \|g[W]\|_1 &\leq 3 r(n) Y \leq 3 y_0 r^2(n) \leq \tau(n,\epsilon) \cdot y_0,
\end{align*}
in which case we set $b_\ell = 2$.  If $\delta_2$ is some small constant, say $\delta_2 \leq \frac{1}{32}$, then for
a fixed iteration $\ell$, we set $b_\ell = 2$ except with probability $\frac{5}{16}$.  Now, applying the Chernoff bound,
we can show that the probability of having more than a $\frac{2}{5}$-fraction of iterations $\ell$ with $b_{\ell} \neq 2$ is at most
an inverse polynomial.  Hence, in this subcase the algorithm outputs $(-1,0)$, except with inverse polynomial probability.

In the second subcase, we have $u_i > \frac{1}{16} \|g[W]\|_1$ and, for all $j \neq i$, $u_j \leq \frac{\epsilon}{128}u_i$.
Then, since $u_i$ is not $\left(1-\frac{\epsilon}{8}\right)$-heavy with respect to $g[W]$, we have:
$$u_j \leq \frac{\epsilon}{128}\cdot u_i \leq \frac{1}{16}(\|g[W]\|_1 - u_i).$$
Hence, we can apply Lemma~\ref{lem:balanced} to the vector $U = (u_1,\ldots,u_{i-1},0,u_{i+1},\ldots,u_n)$ (since
$\|U\|_1 = \|g[W]\|_1 - u_i$, and moreover each entry in $U$ is at most $\frac{1}{16}\|U\|_1$).  Letting
$X' = \sum_{j \neq i}h_j^{\ell}u_j$ and $Y' = \sum_{j \neq i}\bar{h}_j^{\ell}u_j$, we get that:
$$ \Pr\left[\left(X' \leq \frac{1}{4} \cdot \|U\|_1\right) \cup \left(Y' \leq \frac{1}{4} \cdot \|U\|_1\right)\right] \leq \frac{1}{4}.$$
This implies that $X \geq X' > \frac{1}{4}(\|g[W]\|_1 - u_i) \geq \frac{\epsilon}{32}\|g[W]\|_1$ and
$Y \geq Y' > \frac{1}{4}(\|g[W]\|_1 - u_i) \geq \frac{\epsilon}{32}\|g[W]\|_1$.
Moreover, except with probability $2\delta_2$, $y_1$ and $y_0$ are good approximations to $\|g[W_1]\|_1$ and $\|g[W_0]\|_1$, respectively.
Thus, except with probability $\frac{1}{4} + 2\delta_2$, we have:
\begin{align*}
y_0 \leq r(n)Y \leq r(n) \left(1 - \frac{\epsilon}{32}\right) \|g[W]\|_1
&\leq r(n) \left(1 - \frac{\epsilon}{32}\right) \cdot \frac{32}{\epsilon} \cdot X \leq \frac{32r^2(n)}{\epsilon} \cdot y_1 \leq \tau(n,\epsilon) \cdot y_1, \textrm{ and}\\
y_1 \leq r(n)X \leq r(n) \left(1 - \frac{\epsilon}{32}\right) \|g[W]\|_1
&\leq r(n) \left(1 - \frac{\epsilon}{32}\right) \cdot \frac{32}{\epsilon} \cdot Y \leq \frac{32r^2(n)}{\epsilon} \cdot y_0 \leq \tau(n,\epsilon) \cdot y_0.
\end{align*}
This implies that, except with probability $\frac{1}{4} + 2\delta_2$, the algorithm sets $b_\ell = 2$ for each iteration $\ell$.
Applying the Chernoff bound again, we see that the probability of having more than a $\frac{2}{5}$-fraction of iterations $\ell$
with $b_{\ell} \neq 2$ is at most an inverse polynomial.  Thus, in this subcase, the algorithm outputs $(-1,0)$ except with inverse polynomial
probability.

We now consider the last subcase, where $u_i > \frac{1}{16} \|g[W]\|_1$ and there exists $j \neq i$ such that $u_j > \frac{\epsilon}{128}u_i$.
Note that the probability that $i$ and $j$ get mapped to different matrices is given by $\Pr[H_{\ell}(i) \neq H_{\ell}(j)] = \frac{1}{2}$.  Assume
without loss of generality that $H_{\ell}(j) = 1$ (the case that $H_{\ell}(j) = 0$ is symmetric).  In the event
that $i$ and $j$ get mapped to difference matrices and $y_1,y_0$ are good approximations to $\|g[W_1]\|_1, \|g[W_0]\|_1$ respectively,
which happens with probability at least $\frac{1}{2} - 2\delta_2$, we have:
\begin{align*}
y_1 \geq \frac{X}{r(n)} \geq \frac{u_j}{r(n)} \geq \frac{\epsilon}{128r(n)} \cdot u_i &\geq \frac{\epsilon}{128r(n)} \cdot \frac{1}{16} \cdot \|g[W]\|_1 \geq
\frac{\epsilon}{2048r(n)} \cdot Y \geq \frac{\epsilon}{2048r^2(n)} \cdot y_0\\
&\Longrightarrow y_0 \leq \frac{2048r^2(n)}{\epsilon} \cdot y_1 \leq \tau(n,\epsilon) \cdot y_1 \\
y_0 \geq \frac{Y}{r(n)} \geq \frac{u_i}{r(n)} \geq \frac{\epsilon}{128r(n)} \cdot u_i &\geq \frac{\epsilon}{128r(n)} \cdot \frac{1}{16} \cdot \|g[W]\|_1 \geq
\frac{\epsilon}{2048r(n)} \cdot X \geq \frac{\epsilon}{2048r^2(n)} \cdot y_1\\
&\Longrightarrow y_1 \leq \frac{2048r^2(n)}{\epsilon} \cdot y_0 \leq \tau(n,\epsilon) \cdot y_0.
\end{align*}
Thus, except with probability at least $\frac{1}{2} - 2\delta_2$, the algorithm sets $b_\ell = 2$ for each iteration $\ell$.  We
apply the Chernoff bound again to get that $b_{\ell} = 2$ for at least a $\frac{2}{5}$-fraction of iterations, except with
inverse polynomial probability.  Hence, the algorithm outputs $(-1,0)$ except with inverse polynomial probability.

\end{proof}

\begin{lemma}\label{lem:case3}
If the matrix $g[I_H A]$ does not have a key row but has an $\alpha$-heavy row $i_0$, where $\alpha = 1 - \frac{\epsilon}{8}$,
then Algorithm~\ref{alg:key} either returns $(-1,0)$ or returns a $(1 \pm \epsilon)$-approximation of $u_{I_H A,i_0}$ and
the corresponding row $i_0$ with high probability.
\end{lemma}

\begin{proof}
We know there is an $\alpha$-heavy row, but not a key row.  Note that there cannot be more than one $\alpha$-heavy row
for $\alpha = 1 - \frac{\epsilon}{8}$.  If the algorithm returns $(-1,0)$, then the lemma holds
(note the algorithm is allowed to return $(-1,0)$ since there is no key row).  If the algorithm returns
a pair of the form $(i,BA1(A,H))$, we know from Theorem~\ref{BmatrixProof} that the approximation of the
weight of the $\alpha$-heavy row is a $(1 \pm \epsilon)$-approximation of $\|g[W]\|_1$ as long as $BA1$ succeeds,
which happens except with inverse polylogarithmic probability (the argument that the approximation is good follows
similarly as in Lemma~\ref{lem:case1}).  We need only argue that we return the correct index, $i_0$.  Again,
the argument follows similarly as in Lemma~\ref{lem:case1}.  In particular, if $H_{\ell}(i) = b_{\ell}$ for a fixed
iteration $\ell$, then at least one of $y_0,y_1$ is a bad approximation or $H_{\ell}(i_0) = H_{\ell}(i)$,
which happens with probability at most $2\delta_2 + \frac{1}{2}$ (where $\delta_2$ is the failure probability
of $BA2$).  We then apply the Chernoff bound, similarly as before.
\end{proof}

With Lemmas~\ref{lem:case1},~\ref{lem:case2}, and~\ref{lem:case3}, we are done proving that Algorithm~\ref{alg:key} fits
the description of Definition~\ref{HRalgdef}, except with inverse polylogarithmic probability.

\section{Recursive Sketches}\label{recursive}

\noindent Definition of a Cover:

\begin{definition}
A non-empty set $Q \in Pairs_t$, i.e., $Q = \{(i_1, w_1), \dots, (i_t, w_t)\}$ for some $t\in [n]$, is an
$(\alpha, \epsilon)$-cover with respect to the vector $V \in [M]^n$ if the following is true:
\begin{enumerate}
\item $\forall j\in [t]$ $(1-\epsilon)v_{i_j} \le w_{j} \le (1 \pm \epsilon)v_{i_j} $.
\item $\forall i\in [n]$ if $v_i$ is $\alpha$-heavy then $\exists j\in [t]$ such that $i_j = i$.
\end{enumerate}
\end{definition}

\begin{definition}
Let $\mathcal{D}$ be a probability distribution on $Pairs$.
Let $V \in [m]^n$ be a fixed vector.
We say that $\mathcal{D}$ is $\delta$-good with respect to $V$ if for a random element $Q$ of Pairs with distribution D the following is true:
$$
P(Q \text{  is an } (\alpha, \epsilon)\text{-cover of } V ) \ge 1-\delta.
$$
\end{definition}

\noindent Using notation from~\cite{recursive1}, for a vector $V = (v_1,\ldots,v_n)$, we let $|V|$ denote the $L_1$ norm of $V$,
$|V| = \sum_{i=1}^nv_i$.  Consider Algorithm 6 from \cite{recursive1}:

{\begin{algorithm}[H] \caption{Recursive Sum $(D, \epsilon)$}

\begin{enumerate}

\item Generate $\phi=O(\log(n))$ pairwise independent zero-one vectors $H_1, \dots, H_\phi$. Denote by $D_j$ the stream $D_{H_1H_2\dots H_\phi}$

\item Compute, in parallel, $Q_j = HH(D_j, {\epsilon^2\over\phi^3}, \epsilon, {1\over \phi})$

\item If $F_0(V_\phi) > 10^{10}$ then output $0$ and stop. Otherwise, compute precisely $Y_\phi =|V_\phi |$

\item For each $j = \phi-1,\dots, 0$, compute $$Y_j = 2Y_{j+1} - \sum_{i\in Ind(Q_j)} (1-2h_{i}^j)w_{Q_j}(i)$$

\item Output $Y_0$

\end{enumerate}
\label{RecursiveSketch}
\end{algorithm}

\noindent Theorem 4.1 from \cite{recursive1}:

\begin{theorem}
Algorithm~\ref{RecursiveSketch} computes a $(1\pm \epsilon)$-approximation of $|V|$ and errs with probability at most $0.3$.
The algorithm uses $O(\log(n)\mu(n, {1\over \epsilon^2\log^3(n)}, \epsilon, {1\over \log(n)}))$ bits of memory,
where $\mu$ is the space required by the above algorithm $HH$.
\end{theorem}

\end{document}